\documentclass[12pt,a4paper]{article}
\usepackage[margin=2.5cm]{geometry}
\usepackage[english]{babel}
\usepackage[T1]{fontenc}
\usepackage{latexsym,amsmath,amsfonts,amssymb,amsthm,hyperref,tikz}

\numberwithin{equation}{section}
\newcounter{thmcounter}
\numberwithin{thmcounter}{section}
\theoremstyle{definition}
\newtheorem*{acknowledgements}{Acknowledgements}
\newtheorem{definition}[thmcounter]{Definition}
\newtheorem{remark}[thmcounter]{Remark}
\theoremstyle{plain}

\newtheorem{proposition}[thmcounter]{Proposition}
\newtheorem{theorem}[thmcounter]{Theorem}

\def\cA{\mathcal{A}}

\def\cE{\mathcal{E}}

\def\cJ{\mathcal{J}}
\def\cL{\mathcal{L}}

\def\cS{\mathcal{S}}
\def\C{\mathbb{C}}
\def\CP{\mathbb{CP}}
\def\N{\mathbb{N}}
\def\R{\mathbb{R}}
\def\T{\mathbb{T}}
\def\Z{\mathbb{Z}}
\def\1{\mathbf{1}}
\def\b{\mathrm{b}}

\def\SU{\mathrm{SU}}
\def\UN{\mathrm{U}}
\def\diag{\mathrm{diag}}
\def\Re{\mathrm{Re}}
\def\tr{\mathrm{tr}}
\def\ri{\mathrm{i}}
\def\loc{\mathrm{loc}}
\DeclareMathOperator{\sgn}{sgn}
\DeclareMathOperator{\ws}{s}

\begin{document}

\vspace*{.5cm}
\begin{center}
{\Large\bf
Trigonometric and elliptic Ruijsenaars--Schneider systems on the complex projective space}
\end{center}

\medskip
\begin{center}
L.~Feh\'er${}^{a,b}$ and T.F.~G\"orbe${}^a$\\

\bigskip
${}^a$Department of Theoretical Physics, University of Szeged\\
Tisza Lajos krt 84-86, H-6720 Szeged, Hungary\\
e-mail: tfgorbe@physx.u-szeged.hu

\medskip
${}^b$Department of Theoretical Physics, WIGNER RCP, RMKI\\
H-1525 Budapest, P.O.B.~49, Hungary\\
e-mail: lfeher@physx.u-szeged.hu
\end{center}

\medskip
\begin{abstract}
We present a direct construction of compact real forms of the trigonometric and elliptic
$n$-particle Ruijsenaars--Schneider systems whose completed center-of-mass phase space is
the complex projective space $\CP^{n-1}$ with the Fubini--Study symplectic structure.
These systems are labelled by an integer $p\in\{1,\dots,n-1\}$ relative prime to $n$ and a
coupling parameter $y$ varying in a certain punctured interval around $p\pi/n$. Our work
extends Ruijsenaars's pioneering study of compactifications that imposed the restriction
$0<y<\pi/n$, and also builds on an earlier derivation of more general compact trigonometric
systems by Hamiltonian reduction.
\end{abstract}

\newpage
\section{Introduction}
\label{sec:1}

The investigation of integrable systems of particles moving in one spatial dimension started
decades ago and persistently attracts intense attention due to the fascinating mathematics and
diverse physical applications of these systems, as reviewed in \cite{E,Nekr,OP,RuijKup,RuijR,vDVR}.
The Ruijsenaars--Schneider (RS) model \cite{RS,RuijCMP87} occupies a central position in this
family, since many other interesting models of Calogero--Moser--Sutherland and Toda type can
be obtained from it as various limits and analytic continuations \cite{RuijR}. The phase space
of these particle systems is usually the cotangent bundle of the configuration space, which is
never compact due to the infinite range of the canonical momenta. The standard RS Hamiltonian
depends on the momenta $\phi_k$ through the function $\cosh (\phi_k)$, but by analytic
continuation this may be replaced by $\cos (\phi_k)$, which effectively compactifies the
momenta on a circle. If the dependence on the position variables $x_k$ is also through a
periodic function, then the phase space can be taken to be a bounded set. This possibility
was examined in \cite{RIMS95}, where the Hamiltonian
\begin{equation}
H(x,\phi)=\sum_{k=1}^n\cos(\phi_k)\sqrt{\prod_{\substack{j=1\\(j\neq k)}}^n
\bigg[1-\frac{\sin^2y}{\sin^2(x_j-x_k)}\bigg]}
\label{I1}
\end{equation}
containing a real coupling parameter $0<y< \pi/2$ was considered.
Ruijsenaars called this the III$_\b$ system, with III referring to the trigonometric character
of the interaction, as in \cite{OP}, and the suffix standing for `bounded'. (One may also
introduce another real parameter into the III$_\b$ system, by replacing $\phi_k$ say by
$\beta \phi_k$.) The domain of the `angular position variables'
$\{(x_1,\dots, x_n)\}\subset[0,\pi]^n$ must be restricted in such a way that the
Hamiltonian \eqref{I1} is real and smooth. This may be ensured by prescribing
\begin{equation}
x_{i+1}-x_i>y\quad(i=1,\dots,n-1),\quad x_n-x_1<\pi-y,
\label{I2}
\end{equation}
which obviously implies Ruijsenaars's condition
\begin{equation}
0<y<\frac{\pi}{n}.
\label{I3}
\end{equation}
Although the Hamiltonian is then real, its flow is not complete on the naive phase space, because
it may reach the boundary $x_{k+1}-x_k=y$ (with $x_{k+n}\equiv x_k+\pi$) at finite time \cite{RIMS95}.
Completeness of the commuting flows is a crucial property of any bona fide integrable system, but
one cannot directly add the boundary to the phase space because that would not yield
a smooth manifold. One of the seminal results of \cite{RIMS95} is the solution of this
conundrum. In fact, Ruijsenaars constructed a symplectic embedding of the
center-of-mass phase space of the system into the complex projective space $\CP^{n-1}$, such
that the image of the embedding is a dense open submanifold and the Hamiltonian \eqref{I1} as well as its
commuting family extend to smooth functions on the full $\CP^{n-1}$.
As $\CP^{n-1}$ is compact, the corresponding Hamiltonian flows are complete.
The resulting `compactified trigonometric RS system' has been studied at the classical level
in detail \cite{RIMS95}, and after an initial exploration of the rank 1 case \cite{RuijKup},
its quantum mechanical version was also solved \cite{vDV}. These classical systems are self-dual
in the sense that their position and action variables can be exchanged by a canonical transformation
of order 4, somewhat akin to the mapping $(x,\phi) \mapsto (-\phi, x)$ for a free particle, and
their quantum mechanical versions enjoy the bispectral property \cite{RuijKup,vDV}.

The possibility of an analogous compactification
of the elliptic RS system having the Hamiltonian
\begin{equation}
H(x,\phi)=\sum_{k=1}^n\cos(\phi_k)\sqrt{\prod_{\substack{j=1\\(j\neq k)}}^n
\big[\ws(y)^2\big(\wp(y)-\wp(x_j-x_k)\big)\big]}
\label{I5}
\end{equation}
with
 functions
$\wp$ \eqref{wp} and $\ws$ \eqref{sigma}
was pointed out in \cite{RuijKup,RuijR}, but it was not
described in detail.

Even though it was only proved  \cite{RIMS95} that the restrictions \eqref{I2}, \eqref{I3}
are sufficient to allow compactification,  equation \eqref{I3} was customarily mentioned
in the literature  \cite{FK1,GN,RuijKup,RuijR,vDV} as a necessary condition for the systems to make sense.
However, in a recent work \cite{FKl} a completion of the III$_\b$ system on a compact phase space was
obtained for any generic parameter
\begin{equation}
0<y<\pi.
\label{I4}
\end{equation}
The paper \cite{FKl} relied on deriving compactified RS systems in the center-of-mass frame
via reduction of a `free system' on the quasi-Hamiltonian \cite{AMM} double $\SU(n) \times \SU(n)$.
This was achieved by setting
the relevant group-valued moment map equal to the constant matrix
$\mu_0(y)=\diag(e^{2\ri y},\dots,e^{2\ri y},e^{-2(n-1)\ri y})$,
and it makes perfect sense for any (generic) $y$.
The corresponding domain of the position variables depends on $y$ and differs
from the one posited in \eqref{I2}.
The possibility to relax the condition \eqref{I3} on $y$  also appeared in \cite{Bog}.

The principal motivation for our present work comes from the classification of the coupling parameter
$y$ found in \cite{FKl}. Namely, it turned out that the reduction is applicable except for a finite
set of $y$-values, and the rest of the set $(0,\pi)$ decomposes into two subsets, containing so-called
type (i) and type (ii) $y$-values. The `main reduced Hamiltonian' always takes the III$_\b$ form
\eqref{I1} on a dense open subset of the reduced phase space. In the type (i) cases the particles
cannot collide and the action variables of the reduced system naturally engender an isomorphism with
the Hamiltonian toric manifold $\CP^{n-1}$. In type (ii) cases, that exist for any $n>3$, the reduction
constraints admit solutions $(a,b)\in\SU(n)\times\SU(n)$ for which the eigenvalues of $a$ or $b$ are
not all distinct, entailing that the particles of the reduced system can collide.
For a detailed exposition of these succinct statements, the reader may consult \cite{FKl}.
We here only add the remark that the connected domain of the positions
always contains the equal-distance configuration
$x_{k+1} - x_k = \pi/n$ ($\forall k$)  for which
the number of negative factors in each product under the square root in \eqref{I1}
is $2 \lfloor n y/\pi\rfloor$
 if $0<y < \pi/2$ and
 $2 \lfloor n (\pi -y)/\pi \rfloor$  if $\pi/2 < y <\pi$.

This Letter is exclusively concerned with the type (i) cases just mentioned.
Our first goal is to reconstruct the corresponding compactification on $\CP^{n-1}$ using only
direct, elementary methods, i.e., not relying on reduction techniques.
Such construction was not known previously except for the special type (i) cases \eqref{I3}, which we shall
generalize. By doing so, we shall gain a better understanding of the structure of these trigonometric systems.
This part of the Letter fills Sections \ref{sec:2} and \ref{sec:3} that follow.
In Section \ref{sec:4}, we explain that the direct method is applicable to obtain type (i) compactifications
of the elliptic RS system as well. This new result extends the
remarks of Ruijsenaars \cite{RuijKup,RuijR}.

It would have been possible to organize our text differently, starting with the elliptic case
and then recovering the trigonometric systems as a limit.
We opted for first presenting the trigonometric case for the reason that in our hope this makes
the paper easier to understand, and also since this
actually follows our line of research.

Our results lead to several open questions and possible topics for future work that
will be outlined at the end of the paper.

\section{Embedding of the local phase space into $\CP^{n-1}$}
\label{sec:2}

In this section we first recall the local phase space of the III$_\b$ model from \cite{FKl},
and then present its symplectic embedding into $\CP^{n-1}$ in every type (i) case.

The III$_\b$ model can be thought of as $n$ interacting particles on the unit circle
with positions $\delta_k=e^{2\ri x_k}$.
We impose the condition $\prod_{k=1}^n\delta_k=1$, which means that we
work in the `center-of-mass frame',
and parametrize the positions as
\begin{equation}
\delta_1(\xi)=e^{\frac{2\ri}{n}\sum_{j=1}^nj\xi_j},\qquad
\delta_{k}(\xi)=e^{2\ri\xi_{k-1}}\delta_{k-1}(\xi),\quad k=2,\dots,n,
\label{delta}
\end{equation}
where $\xi$ belongs to a certain open subset $\cA_y^+$ inside the `Weyl alcove'
\begin{equation}
\cA=\{\xi\in\R^n\mid\xi_k\geq 0\ (k=1,\dots,n),\ \xi_1+\dots+\xi_n=\pi\}.
\label{Weyl}
\end{equation}
Note that $\cA$ is a simplex in the $(n-1)$-dimensional affine space
\begin{equation}
E=\{\xi\in\R^n\mid\xi_1+\dots+\xi_n=\pi\}.
\label{E}
\end{equation}
The local phase space can be described as the product manifold
\begin{equation}
P_y^\loc=\{(\xi,e^{\ri\theta})\mid\xi\in\cA_y^+,\ e^{\ri\theta}\in\T^{n-1}\},
\label{P-loc}
\end{equation}
where $\T^{n-1}$ is the $(n-1)$-torus, equipped with the standard symplectic form
\begin{equation}
\omega^\loc=\sum_{k=1}^{n-1}d\theta_k\wedge d\xi_k.
\label{om-loc}
\end{equation}
The dynamics is governed by the Hamiltonian
\begin{equation}
H_y^\loc(\xi,\theta)=\sum_{j=1}^n\cos(\theta_j-\theta_{j-1})\sqrt{ \prod_{m=j+1}^{j+n-1}
\biggl[ 1-\frac{\sin^2y}{\sin^2(\sum_{k=j}^{m-1}\xi_k)}\biggr]}.
\label{H_y^loc}
\end{equation}
Here, $\theta_0=\theta_n=0$ have been introduced and the indices are understood
modulo $n$, i.e.,
\begin{equation}
\xi_{m+n}=\xi_m, \quad\forall m.
\label{perconv}
\end{equation}
The product under the square
root is positive for every $\xi\in\cA_y^+$, and thus $H_y^\loc\in C^\infty(P_y^\loc)$.
This model was considered in \cite{FKl} for any $y$ chosen from the interval $(0,\pi)$
except the excluded values that satisfy $e^{2\ri my}=1$ for some $m=1,\dots,n$.

According to \cite{FKl}, there are two different kinds of intervals for $y$ to be in,
named type (i) and (ii). The type (i) couplings can be described as follows. For a
fixed positive integer $n\geq 2$, choose $p\in\{1,\dots,n-1\}$ to be a coprime to
$n$, i.e., $\gcd(n,p)=1$, and let $q$ denote the multiplicative inverse of $p$ in
the ring $\Z_n$, that is $pq\equiv 1\pmod{n}$. Then the parameter $y$ can take its
values according to either
\begin{equation}
\bigg(\frac{p}{n}-\frac{1}{nq}\bigg)\pi<y<\frac{p\pi}{n}
\qquad\text{or}\qquad
\frac{p\pi}{n}<y<\bigg(\frac{p}{n}+\frac{1}{(n-q)n}\bigg)\pi.
\label{typeI-y}
\end{equation}
For such a type (i) parameter $y$, the local configuration space $\cA_y^+$
is the interior of a simplex $\cA_y$ in $E$ \eqref{E}
bounded by the hyperplanes
\begin{equation}
\xi_j+\dots+\xi_{j+p-1}=y,\quad j=1,\dots,n,
\label{hyperplanes}
\end{equation}
where \eqref{perconv} is understood.
To give a more detailed description of $\cA_y$, we introduce
\begin{equation}
M=p\pi-ny,
\label{M}
\end{equation}
and note that \eqref{typeI-y} gives $M>0$ and $M<0$, respectively.
Then any $\xi\in\cA_y$ must satisfy
\begin{equation}
\sgn(M)(\xi_j+\dots+\xi_{j+p-1}-y)\geq 0,\quad j=1,\dots,n.
\label{3.11}
\end{equation}
In terms of the particle coordinates $x_k$, which are ordered as $x_{k+1} \geq x_k$ and
extended by the convention $x_{k+n} = x_k + \pi$, the above condition says that
\begin{equation}
x_{j+p}-x_j\geq y\quad\text{if}\ M>0
\quad\text{and}\quad
x_{j+p}-x_j\leq y\quad\text{if}\ M<0
\label{3.12}
\end{equation}
for every $j$.
Therefore the distances of the $p$-th neighbouring particles on the circle are constrained.
The $n$ vertices of the simplex $\cA_y$ are explicitly given in \cite{FKl} (Proposition 11
and Lemma 8 \emph{op. cit.}). Every vertex and thus $\cA_y$ itself lies inside the larger
simplex $\cA$ \eqref{Weyl}, entailing that $x_{j+1}-x_j$ possesses a positive lower bound
in each type (i) case.

The type (ii) cases correspond to those admissible $y$-values that do not satisfy
\eqref{typeI-y} for any $p$ relative prime to $n$. In such cases $\cA_y^+$ has a different
structure \cite{FKl}. Type (ii) cases exist for every $n\geq 4$. See Figure \ref{figure:1}
for an illustration.

\begin{figure}[h!]
\centering
\begin{tikzpicture}
\def\s{.9\textwidth}
\def\r{.2em}
\draw(1em,1em) node{$n=4$};
\draw (0,0)--({\s/3},0) ({2*\s/3},0)--({\s},0);
\draw[dashed] ({\s/3},0)--({2*\s/3},0);
\draw[black,fill=white]
(0,0) circle(\r) node[below,yshift=-1mm]{$0$}
(\s/4,0) circle(\r) node[below,yshift=-1mm]{$\displaystyle\frac{1}{4}$}
(\s/3,0) circle(\r) node[below,yshift=-1mm]{$\displaystyle\frac{1}{3}$}
(\s/2,0) circle(\r) node[below,yshift=-1mm]{$\displaystyle\frac{1}{2}$}
(2*\s/3,0) circle(\r) node[below,yshift=-1mm]{$\displaystyle\frac{2}{3}$}
(3*\s/4,0) circle(\r) node[below,yshift=-1mm]{$\displaystyle\frac{3}{4}$}
(\s,0) circle(\r) node[below,yshift=-1mm]{$1$};
\draw(1em,-3em) node{$n=5$};
\draw (0,-4em)--(\s/4,-4em) (\s/3,-4em)--(2*\s/3,-4em) (3*\s/4,-4em)--(\s,-4em);
\draw[dashed] (\s/4,-4em)--(\s/3,-4em) (2*\s/3,-4em)--(3*\s/4,-4em);
\draw[black,fill=white]
(0,-4em) circle(\r) node[below,yshift=-1mm]{$0$}
(\s/5,-4em) circle(\r) node[below,yshift=-1mm]{$\displaystyle\frac{1}{5}$}
(\s/4,-4em) circle(\r) node[below,yshift=-1mm]{$\displaystyle\frac{1}{4}$}
(\s/3,-4em) circle(\r) node[below,yshift=-1mm]{$\displaystyle\frac{1}{3}$}
(2*\s/5,-4em) circle(\r) node[below,yshift=-1mm]{$\displaystyle\frac{2}{5}$}
(\s/2,-4em) circle(\r) node[below,yshift=-1mm]{$\displaystyle\frac{1}{2}$}
(3*\s/5,-4em) circle(\r) node[below,yshift=-1mm]{$\displaystyle\frac{3}{5}$}
(2*\s/3,-4em) circle(\r) node[below,yshift=-1mm]{$\displaystyle\frac{2}{3}$}
(3*\s/4,-4em) circle(\r) node[below,yshift=-1mm]{$\displaystyle\frac{3}{4}$}
(4*\s/5,-4em) circle(\r) node[below,yshift=-1mm]{$\displaystyle\frac{4}{5}$}
(\s,-4em) circle(\r) node[below,yshift=-1mm]{$1$};
\draw(1em,-7em) node{$n=6$};
\draw (0,-8em)--(\s/5,-8em) (4*\s/5,-8em)--(\s,-8em);
\draw[dashed] (\s/5,-8em)--(4*\s/5,-8em);
\draw[black,fill=white]
(0,-8em) circle(\r) node[below,yshift=-1mm]{$0$}
(\s/6,-8em) circle(\r) node[below,yshift=-1mm]{$\displaystyle\frac{1}{6}$}
(\s/5,-8em) circle(\r) node[below,yshift=-1mm]{$\displaystyle\frac{1}{5}$}
(\s/4,-8em) circle(\r) node[below,yshift=-1mm]{$\displaystyle\frac{1}{4}$}
(\s/3,-8em) circle(\r) node[below,yshift=-1mm]{$\displaystyle\frac{1}{3}$}
(2*\s/5,-8em) circle(\r) node[below,yshift=-1mm]{$\displaystyle\frac{2}{5}$}
(\s/2,-8em) circle(\r) node[below,yshift=-1mm]{$\displaystyle\frac{1}{2}$}
(3*\s/5,-8em) circle(\r) node[below,yshift=-1mm]{$\displaystyle\frac{3}{5}$}
(2*\s/3,-8em) circle(\r) node[below,yshift=-1mm]{$\displaystyle\frac{2}{3}$}
(3*\s/4,-8em) circle(\r) node[below,yshift=-1mm]{$\displaystyle\frac{3}{4}$}
(4*\s/5,-8em) circle(\r) node[below,yshift=-1mm]{$\displaystyle\frac{4}{5}$}
(5*\s/6,-8em) circle(\r) node[below,yshift=-1mm]{$\displaystyle\frac{5}{6}$}
(\s,-8em) circle(\r) node[below,yshift=-1mm]{$1$};
\draw(1em,-11em) node{$n=7$};
\draw (0,-12em)--(\s/6,-12em) (\s/4,-12em)--(\s/3,-12em) (2*\s/5,-12em)--(3*\s/5,-12em)
(2*\s/3,-12em)--(3*\s/4,-12em) (5*\s/6,-12em)--(\s,-12em);
\draw[dashed] (\s/6,-12em)--(\s/4,-12em) (\s/3,-12em)--(2*\s/5,-12em)
(3*\s/5,-12em)--(2*\s/3,-12em) (3*\s/4,-12em)--(5*\s/6,-12em);
\draw[black,fill=white]
(0,-12em) circle(\r) node[below,yshift=-1mm]{$0$}
(\s/7,-12em) circle(\r) node[below,yshift=-1mm]{$\displaystyle\frac{1}{7}$}
(\s/6,-12em) circle(\r) node[below,yshift=-1mm]{$\displaystyle\frac{1}{6}$}
(\s/5,-12em) circle(\r) node[below,yshift=-1mm]{$\displaystyle\frac{1}{5}$}
(\s/4,-12em) circle(\r) node[below,yshift=-1mm]{$\displaystyle\frac{1}{4}$}
(2*\s/7,-12em) circle(\r) node[below,yshift=-1mm]{$\displaystyle\frac{2}{7}$}
(\s/3,-12em) circle(\r) node[below,yshift=-1mm]{$\displaystyle\frac{1}{3}$}
(2*\s/5,-12em) circle(\r) node[below,yshift=-1mm]{$\displaystyle\frac{2}{5}$}
(3*\s/7,-12em) circle(\r) node[below,yshift=-1mm]{$\displaystyle\frac{3}{7}$}
(\s/2,-12em) circle(\r) node[below,yshift=-1mm]{$\displaystyle\frac{1}{2}$}
(4*\s/7,-12em) circle(\r) node[below,yshift=-1mm]{$\displaystyle\frac{4}{7}$}
(3*\s/5,-12em) circle(\r) node[below,yshift=-1mm]{$\displaystyle\frac{3}{5}$}
(2*\s/3,-12em) circle(\r) node[below,yshift=-1mm]{$\displaystyle\frac{2}{3}$}
(5*\s/7,-12em) circle(\r) node[below,yshift=-1mm]{$\displaystyle\frac{5}{7}$}
(3*\s/4,-12em) circle(\r) node[below,yshift=-1mm]{$\displaystyle\frac{3}{4}$}
(4*\s/5,-12em) circle(\r) node[below,yshift=-1mm]{$\displaystyle\frac{4}{5}$}
(5*\s/6,-12em) circle(\r) node[below,yshift=-1mm]{$\displaystyle\frac{5}{6}$}
(6*\s/7,-12em) circle(\r) node[below,yshift=-1mm]{$\displaystyle\frac{6}{7}$}
(\s,-12em) circle(\r) node[below,yshift=-1mm]{$1$};
\end{tikzpicture}
\caption{The range of $y/\pi$ for $n=4,5,6,7$. The displayed numbers are excluded values.
Admissible values of $y$ form intervals of type (i) (solid) and type (ii) (dashed) couplings.}
\label{figure:1}
\end{figure}
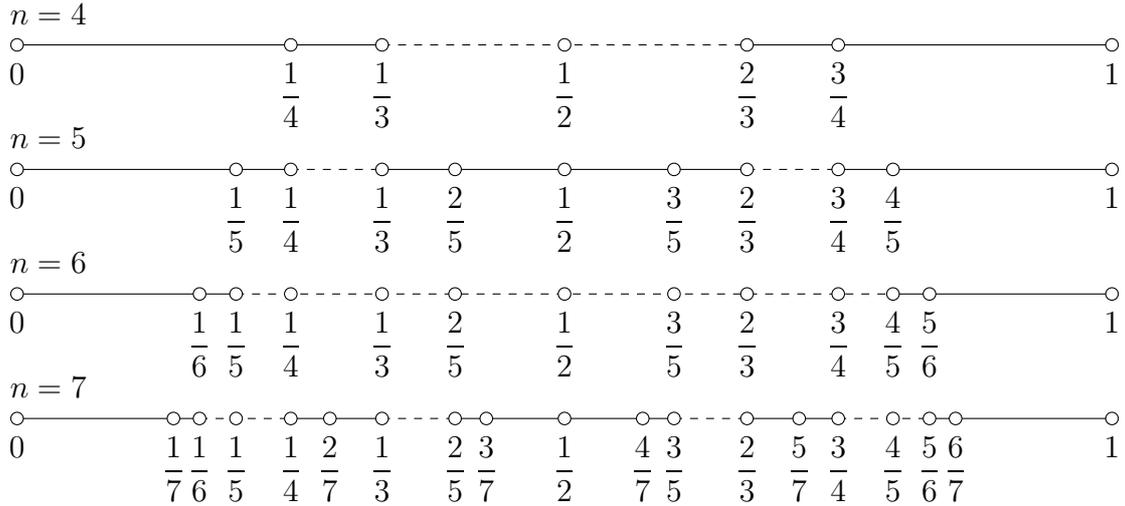

We further continue with the assumption that $y$ satisfies \eqref{typeI-y}.
Motivated by \cite{RIMS95,FK1}, we now introduce the map
\begin{equation}
\cE\colon\cA_y^+\times\T^{n-1}\to\C^n,\quad
(\xi,e^{\ri\theta})\mapsto (u_1,\dots,u_n)
\label{Emap}
\end{equation}
with the complex coordinates having the squared absolute values
\begin{equation}
|u_j|^2=\sgn(M)(\xi_j+\dots+\xi_{j+p-1}-y),\quad j=1,\dots,n,
\label{u-abs-squared}
\end{equation}
and the arguments
\begin{equation}
\arg(u_j)=\sgn(M)\sum_{k=1}^{n-1}\Omega_{j,k}\theta_k,\quad j=1,\dots,n-1,
\qquad\arg(u_n)=0,
\label{argu}
\end{equation}
where the $\Omega_{j,k}$ ($j,k=1,\dots,n-1$) are integers chosen in such a way that
\begin{equation}
\cE^\ast\bigg(\ri\sum_{j=1}^nd\bar u_j\wedge du_j\bigg)
=\sum_{k=1}^{n-1}d\theta_k\wedge d\xi_k.
\label{cE-sympl}
\end{equation}
In order for \eqref{cE-sympl} to be achieved $\Omega$ has to be the inverse transpose of
the $(n-1)\times(n-1)$ coefficient matrix of $\xi_1,\dots,\xi_{n-1}$ extracted from eqs.
\eqref{u-abs-squared} by applying $\xi_1+\dots+\xi_n=\pi$. In other words, the squared
absolute values $|u_j|^2$ are written as
\begin{equation}
|u_j|^2=\begin{cases}
\sgn(M)\big(\sum_{k=1}^{n-1}A_{j,k}\xi_k-y\big),&\text{if}\ 1\leq j\leq n-p,\\
\sgn(M)\big(\sum_{k=1}^{n-1}A_{j,k}\xi_k-y+\pi\big),&\text{if}\ n-p<j\leq n-1,
\end{cases}
\label{uabs}
\end{equation}
where $A$ stands for the above-mentioned coefficient matrix, which has the components
\begin{equation}
A_{j,k}=\begin{cases}
+1,&\text{if}\ 1\leq j\leq n-p\ \text{and}\ j\leq k<j+p,\\
-1,&\text{if}\ n-p<j\leq n-1\ \text{and}\ j+p-n\leq k<j,\\
0,&\text{otherwise}.
\end{cases}
\label{A_j,k}
\end{equation}
A close inspection of the structure of $A$ reveals that
\begin{equation}
\det(A)=(-1)^{(n-p)(p-1)}\prod_{j=1}^{n-p}A_{j,j+p-1}\prod_{k=1}^{p-1}A_{n-p+k,k}
=(-1)^{(n-p+1)(p-1)}=+1,
\label{det(A)}
\end{equation}
therefore $\Omega=(A^{-1})^\top$ exists and consists of integers, as required in \eqref{argu}.
Next, we give $\Omega$ explicitly.

\begin{proposition}
\label{prop:2.1}
The transpose of the inverse of the matrix $A$ \eqref{A_j,k} can be written as
\begin{equation}
\Omega=B-C,
\label{B-C}
\end{equation}
where $B$ is a $(0,1)$-matrix of size $(n-1)$ with zeros along certain diagonals given by
\begin{equation}
B_{m,k}=\begin{cases}
0,&\text{if}\ k-m\equiv \ell p\pmod{n}\ \text{for some}\
\ell\in\{1,\dots,n-q\},\\
1,&\text{otherwise},
\end{cases}
\label{B_m,k}
\end{equation}
and $C$ is also a binary matrix of size $(n-1)$ with zeros along columns given by
\begin{equation}
C_{m,k}=\begin{cases}
0,&\text{if}\ k\equiv \ell p\pmod{n}\ \text{for some}\
\ell\in\{1,\dots,n-q\},\\
1,&\text{otherwise}.
\end{cases}
\label{C_m,k}
\end{equation}
\end{proposition}
\begin{proof}
We start by presenting a useful auxiliary statement.
Let us introduce the subsets $S$ and $S_i$ of the ring $\Z_n$ as
\begin{equation}
S=\{\ell p\ (\bmod\ n)\mid\ell=1,\dots,n-q\},
\quad
S_i=\{i+\ell\ (\bmod\ n)\mid\ell=0,\dots,p-1\},
\end{equation}
for any $i\in\Z_n$.
Then define $I_i\in\N$ to be the number of elements in the intersection $S_i\cap S$.
Notice that $i\in S$ if and only if $(i+p)\in S$ except for $i\equiv(n-1)\equiv(n-q)p\pmod{n}$,
for which $(n-1)+p\equiv(n-q+1)p\pmod{n}$ does not belong to $S$.
It follows that
\begin{equation}
I_1=\dots=I_{n-1}=I_n+1.
\label{shiftid}
\end{equation}

Our aim is to show that $(A\Omega^\top)_{j,m}=\delta_{j,m}$ ($\forall j,m$)
with $\Omega$ defined by \eqref{B-C}-\eqref{C_m,k}. First, by the formula of $A$
\eqref{A_j,k} for any $1\leq j\leq n-p$ and $1\leq m\leq n-1$ we have
\begin{equation}
(A\Omega^\top)_{j,m}
=\sum_{k=1}^{n-1}A_{j,k}\Omega_{m,k}
=\sum_{k=j}^{j+p-1}\Omega_{m,k}
=\sum_{k=j}^{j+p-1}(B_{m,k}-C_{m,k}).
\label{}
\end{equation}
The definition of the matrices $B$ \eqref{B_m,k} and $C$ \eqref{C_m,k} gives directly that
\begin{equation}
\sum_{k=j}^{j+p-1}B_{m,k}=p-I_{j-m},
\qquad
\sum_{k=j}^{j+p-1}C_{m,k}=p-I_j.
\label{}
\end{equation}
By using \eqref{shiftid}, this readily implies that
$(A\Omega^\top)_{j,m}=\delta_{j,m}$ holds for the case at hand.

Second, for any $n-p<j\leq n-1$ and $1\leq m\leq n-1$ we have
\begin{equation}
(A\Omega^\top)_{j,m}
=\sum_{k=1}^{n-1}A_{j,k}\Omega_{m,k}
=\sum_{k=j+p-n}^{j-1}(-1)\Omega_{m,k}
=\sum_{k=j+p-n}^{j-1}(C_{m,k}-B_{m,k}).
\label{}
\end{equation}
From this point on the reasoning is quite similar to the previous case, and we obtain that
$(A\Omega^\top)_{j,m}=\delta_{j,m}$ always holds.
\end{proof}

To enlighten the geometric meaning of the map $\cE$ \eqref{Emap}, notice from
\eqref{u-abs-squared} that
\begin{equation}
\sum_{j=1}^n|u_j|^2=\sgn(M)\big(p(\xi_1+\dots+\xi_n)-ny\big)
=\sgn(M)\big(p\pi-ny\big)=|M|.
\label{}
\end{equation}
Then represent the complex projective space $\CP^{n-1}$ as
\begin{equation}
\CP^{n-1}=S_{|M|}^{2n-1}/\UN(1)
\label{coset}
\end{equation}
with
\begin{equation}
S_{|M|}^{2n-1}=\{(u_1,\dots,u_n)\in\C^n\mid|u_1|^2+\dots+|u_n|^2=|M|\}.
\label{sphere}
\end{equation}
Correspondingly, let
\begin{equation}
\pi_{|M|}\colon S_{|M|}^{2n-1}\to\CP^{n-1}
\label{}
\end{equation}
denote the natural projection and equip $\CP^{n-1}$ with the rescaled Fubini--Study
symplectic form $|M|\omega_{\text{FS}}$ characterized by the relation
\begin{equation}
\pi_{|M|}^\ast(|M|\omega_{\text{FS}})=\ri\sum_{j=1}^nd\bar u_j\wedge du_j,
\label{2.32}
\end{equation}
where the $u_j$'s are regarded as functions on $S^{2n-1}_{\vert M\vert }$.
It is readily seen from the definitions that the map
\begin{equation}
\pi_{|M|}\circ\cE\colon\cA_y^+\times\T^{n-1}\to\CP^{n-1}
\label{pi+cE}
\end{equation}
is smooth, injective and its image is the open submanifold for which
$\prod_{j=1}^n \vert u_j \vert^2 \neq 0$.
Equations \eqref{om-loc}, \eqref{cE-sympl} and \eqref{2.32}  together imply the symplectic property
\begin{equation}
(\pi_{|M|}\circ\cE)^\ast (\vert M\vert \omega_{\text{FS}}) = \omega^\loc,
\end{equation}
from which it follows that this map is an \emph{embedding}.

To summarize, in this section we have constructed the
symplectic diffeomorphism $\pi_{|M|}\circ\cE$ between the local phase space
$P_y^\loc$ \eqref{P-loc}
and the dense open submanifold of $\CP^{n-1}$ on which the product of the homogeneous
coordinates is nowhere zero.
If desired, the explicit formula of the smooth inverse mapping can be easily found as well.

\section{Global extension of the trigonometric Lax matrix}
\label{sec:3}

It was proved in \cite{FKl} with the aid of quasi-Hamiltonian reduction that the global
phase space of the III$_\b$ model is $\CP^{n-1}$ for the type (i) couplings, which we
continue to consider. Here, we utilize the symplectic embedding \eqref{pi+cE} to
construct a global Lax matrix on $\CP^{n-1}$ explicitly, starting from the local
RS Lax matrix defined on $\cA_y^+\times\T^{n-1}$. This issue was not investigated
previously except for the $p=1$ case of \eqref{typeI-y}, see \cite{RIMS95,FK1,FKl}.

The local Lax matrix $L_y^\loc(\xi,e^{\ri \theta})\in\SU(n)$ used in \cite{FKl} contains the
 trigonometric Cauchy matrix $C_y$ given with the help of \eqref{delta} by
\begin{equation}
C_y(\xi)_{j,\ell}=\frac{e^{\ri y}-e^{-\ri y}}
{e^{\ri y}\delta_j(\xi)^{1/2}\delta_\ell(\xi)^{-1/2}
-e^{-\ri y}\delta_j(\xi)^{-1/2}\delta_\ell(\xi)^{1/2}}.
\label{C_y}
\end{equation}
Thanks to the relation $\delta_k(\xi)= e^{2\ri x_k}$, this is equivalent to
\begin{equation}
C_y(\xi)_{j,\ell}=\frac{\sin(y)}{\sin(x_j-x_\ell+y)}.
\label{C_y-2}
\end{equation}
Then we have
\begin{equation}
L_y^\loc(\xi,e^{\ri\theta})_{j,\ell}=
C_y(\xi)_{j,\ell}v_j(\xi,y)v_\ell(\xi,-y)\rho(\theta)_\ell,
\qquad
\forall (\xi, e^{\ri \theta})\in \cA_y^+\times\T^{n-1},
\label{L_y^loc}
\end{equation}
where $\rho(\theta)_\ell=e^{\ri(\theta_{\ell-1}-\theta_\ell)}$
(applying $\theta_0=\theta_n=0$) and
\begin{equation}
v_\ell(\xi,\pm y)=\sqrt{z_\ell(\xi,\pm y)}\quad\text{with}\quad
z_\ell(\xi,\pm y)=\sgn(\sin(ny))\prod_{m=\ell+1}^{\ell+n-1}
\frac{\sin(\sum_{k=\ell}^{m-1}\xi_k\mp y)}{\sin(\sum_{k=\ell}^{m-1}\xi_k)}.
\label{z_ell(xi,pmy)}
\end{equation}
A key point \cite{FKl} (which is detailed below) is that $z_\ell(\xi,\pm y)$ is
positive for any $\xi \in \cA_y^+$. We note for clarity that $z_\ell$ and $v_\ell$
above differ from those in \cite{FKl} by a harmless multiplicative constant,
and also mention that $L_y^\loc$ is a specialization of (a similarity transform of)
the standard RS Lax matrix \cite{RuijR}.

The spectral invariants of $L_y^\loc$ \eqref{L_y^loc} yield a Poisson
commuting family of functional dimension $(n-1)$ \cite{RuijR,FKl}, containing
the Hamiltonian $H_y^\loc$ \eqref{H_y^loc} due to the equation
\begin{equation}
\Re\big(\tr L_y^\loc(\xi,e^{\ri\theta})\big)=H_y^\loc(\xi,\theta).
\label{Re-tr-L_y^loc}
\end{equation}
There are two important observations to be made here. First, for each
$1\leq\ell\leq n$, there is only one factor in $z_\ell(\xi,\pm y)$
\eqref{z_ell(xi,pmy)} that (up to sign) contains the sine of the squared
absolute value \eqref{u-abs-squared} of one of the complex variables in
its numerator:
\begin{itemize}
\item For $z_\ell(\xi,y)$, it is the factor corresponding to $m=\ell+p$,
whose numerator is
\begin{equation}
\sgn(M)\sin(|u_\ell|^2).
\label{+y-num}
\end{equation}
\item For $z_\ell(\xi,-y)$, it is the factor with $m=\ell+n-p$, whose the numerator is either
\begin{equation}
\sin(\pi-\sgn(M)|u_{\ell+n-p}|^2)=\sgn(M)\sin(|u_{\ell+n-p}|^2),
\quad\text{if}\ 1\leq\ell\leq p,
\label{-y-num-1}
\end{equation}or
\begin{equation}
\sin(\pi-\sgn(M)|u_{\ell-p}|^2)=\sgn(M)\sin(|u_{\ell-p}|^2),
\quad\text{if}\ p<\ell\leq n.
\label{-y-num-2}
\end{equation}
Here we made use of $\xi_1+\dots+\xi_n=\pi$, $\sin(\pi-\alpha)=\sin(\alpha)$
and $\sin(-\alpha)=-\sin(\alpha)$.
\end{itemize}
Second, the $(p-1)$ factors in $z_\ell(\xi,\pm y)$ with
$m<\ell+p$ and $m>\ell+n-p$, respectively, are strictly negative and the
factors corresponding to $m>\ell+p$ and $m<\ell+n-p$, respectively, are
strictly positive for all $\xi$ in the \emph{closed} simplex $\cA_y$.
In particular, for any $\xi\in\cA_y^+$ the sign of the $\xi$-dependent product in
\eqref{z_ell(xi,pmy)} equals $(-1)^{p-1}\sgn(M)=\sgn(\sin(ny))$, and therefore
\begin{equation}
z_\ell(\xi,\pm y)\geq 0,\quad\forall\xi\in\cA_y,\quad\ell=1,\dots,n.
\label{z-ell-pos}
\end{equation}
We saw that $z_\ell$ can only vanish due to the numerators \eqref{+y-num} and
\eqref{-y-num-1}, \eqref{-y-num-2}, respectively.
Consequently, in \eqref{z_ell(xi,pmy)} the positive square
root of $z_\ell(\xi,\pm y)$ can be taken for any $\xi\in\cA_y^+$.

Now notice that, for all $\xi\in \cA_y^+$, we have
\begin{equation}
v_j(\xi,y)=|u_j|w_j(\xi,y),\quad 1\leq j\leq n,
\label{v_ell(xi,y)}
\end{equation}
where the $w_j(\xi,y)$ are positive and smooth functions of the form
\begin{equation}
w_j(\xi,y)=\bigg[
\frac{\sin(|u_j|^2)}{|u_j|^2}
\frac{(-1)^{p-1}}{\sin(\sum_{k=j}^{j+p-1}\xi_k)}
\prod_{\substack{m=j+1\\(m\neq j+p)}}^{j+n-1}
\frac{\sin(\sum_{k=j}^{m-1}\xi_k-y)}{\sin(\sum_{k=j}^{m-1}\xi_k)}
\bigg]^{\tfrac{1}{2}}.
\label{w_j(xi,y)}
\end{equation}
Similarly, we have
\begin{equation}
v_\ell(\xi,-y)=\begin{cases}
|u_{\ell+n-p}|w_\ell(\xi,-y),&\text{if}\ 1\leq\ell\leq p,\\
|u_{\ell-p}|w_\ell(\xi,-y),&\text{if}\ p<\ell\leq n
\end{cases}
\label{v_ell(xi,-y)}
\end{equation}
with the positive and smooth functions
\begin{equation}
w_\ell(\xi,-y)=\bigg[
\frac{\sin(|u_{\ell+n-p}|^2)}{|u_{\ell+n-p}|^2}
\frac{(-1)^{p-1}}{\sin(\sum_{k=\ell}^{\ell+n-p-1}\xi_k)}
\prod_{\substack{m=\ell+1\\(m\neq\ell+n-p)}}^{\ell+n-1}
\frac{\sin(\sum_{k=\ell}^{m-1}\xi_k+y)}{\sin(\sum_{k=\ell}^{m-1}\xi_k)}
\bigg]^{\tfrac{1}{2}}
\label{w_ell(xi,-y)-1}
\end{equation}
for $1\leq\ell\leq p$, and
\begin{equation}
w_\ell(\xi,-y)=\bigg[
\frac{\sin(|u_{\ell-p}|^2)}{|u_{\ell-p}|^2}
\frac{(-1)^{p-1}}{\sin(\sum_{k=\ell}^{\ell+n-p-1}\xi_k)}
\prod_{\substack{m=\ell+1\\(m\neq\ell+n-p)}}^{\ell+n-1}
\frac{\sin(\sum_{k=\ell}^{m-1}\xi_k+ y)}{\sin(\sum_{k=\ell}^{m-1}\xi_k)}
\bigg]^{\tfrac{1}{2}}
\label{w_ell(xi,-y)-2}
\end{equation}
for $p<\ell\leq n$.

The relation \eqref{uabs} allows us to express the $\xi_k$ in terms
of the complex variables for $k=1,\dots,n-1$ as
\begin{equation}
\xi_k(u)=
\sum_{j=1}^{n-1}\Omega_{j,k}\big(\sgn(M)|u_j|^2+c_j\big),\quad\text{with}\
c_j=\begin{cases}y,&\text{if}\ 1\leq j\leq n-p,\\
y-\pi,&\text{if}\ n-p<j\leq n-1,
\end{cases}
\label{xi(u)}
\end{equation}
and $\xi_n(u)=\pi-\xi_1(u)-\dots-\xi_{n-1}(u)$.
These formulas extend to $\UN(1)$-invariant smooth functions on $S_{|M|}^{2n-1}$,
which represent smooth functions on $\CP^{n-1}$ on account of \eqref{coset}.
By applying these, the above expressions $w_j(\xi(u), \pm y)$ $(j=1,\dots,n)$
\emph{give rise to smooth functions on $\CP^{n-1}$}.

\begin{definition}
\label{def:3.1}
By setting $\theta_k=0$ $(\forall k)$ in the local Lax matrix $L_y^\loc$
\eqref{L_y^loc} with $y$ \eqref{typeI-y}, we define the functions
$\Lambda_{j,\ell}^y\colon\cA_y^+\to\R$ ($j,\ell=1,\dots,n$) via the equations
\begin{equation}
\Lambda_{j,j+p}^y(\xi)=L_y^\loc(\xi,\1_{n-1})_{j,j+p},\quad 1\leq j\leq n-p,
\label{L_y^loc_j,j+p}
\end{equation}
\begin{equation}
\Lambda_{j,j+p-n}^y(\xi)=L_y^\loc(\xi,\1_{n-1})_{j,j+p-n},\quad n-p<j\leq n,
\label{L_y^loc_j,j+p-n}
\end{equation}
\begin{equation}
\Lambda_{j,\ell}^y(\xi)=L_y^\loc(\xi,\1_{n-1})_{j,\ell}(|u_j||u_{\ell+n-p}|)^{-1},\quad
1\leq j\leq n,\ 1\leq\ell\leq p\quad (\ell\neq j+p-n),
\label{L_y^loc_j,ell-1}
\end{equation}
\begin{equation}
\Lambda_{j,\ell}^y(\xi)=L_y^\loc(\xi,\1_{n-1})_{j,\ell}(|u_j||u_{\ell-p}|)^{-1},\quad
1\leq j\leq n,\ p<\ell\leq n\quad (\ell\neq j+p).
\label{L_y^loc_j,ell-2}
\end{equation}
\end{definition}

The foregoing results lead to explicit formulas for $\Lambda_{j,\ell}^y$
(see Appendix \ref{sec:A}). Using the identification \eqref{coset}
and \eqref{xi(u)}, it is readily seen that the $\Lambda^y_{j,\ell}(\xi(u))$
given by Definition \ref{def:3.1} extend to smooth functions on $\CP^{n-1}$.

\begin{remark}
\label{rem:3.2}
The explicit formulas of $\Lambda_{j,\ell}^y(\xi(u))$ contain products of square roots
of strictly positive functions depending on $|u_k|^2\in C^\infty(S^{2n-1}_{|M|})^{\UN(1)}$
for $k=1,\dots,n$. In particular, they contain the square root of the function $J$ given by
\begin{equation}
J(|u_k|^2)=\frac{\sin(|u_k|^2)}{|u_k|^2},
\label{J1}
\end{equation}
which remains smooth (even real-analytic) at $|u_k|^2=0$ and is positive since
we have $0\leq|u_k|^2\leq|M|<\pi$. Indeed, $|M|<\pi/q$ and $|M|<\pi/(n-q)$,
respectively, for the two intervals of the type (i) couplings in \eqref{typeI-y}.
\end{remark}

The above observations allow us to introduce the following functions,
which will be used to construct the global Lax matrix.
\begin{definition}
\label{def:3.3}
For $M>0$ \eqref{M}, define the smooth functions $L_{j,\ell}^{y,+}\colon\CP^{n-1}\to\C$ by
\begin{align}
L_{j,\ell}^{y,+}\circ \pi_{\vert M\vert}(u)&=\begin{cases}
\Lambda_{j,j+p}^y(\xi(u)),&\text{if}\ 1\leq j\leq n-p,\ \ell=j+p,\\
\Lambda_{j,j+p-n}^y(\xi(u)),&\text{if}\ n-p<j\leq n,\ \ell=j+p-n,\\
\bar u_ju_{\ell+n-p}\Lambda_{j,\ell}^y(\xi(u)),&\text{if}\ 1\leq j\leq n,\ 1\leq\ell\leq p,\ \ell\neq j+p-n,\\
\bar u_ju_{\ell-p}\Lambda_{j,\ell}^y(\xi(u)),&\text{if}\ 1\leq j\leq n,\ p<\ell\leq n,\ \ell\neq j+p,
\end{cases}
\label{L^y+}
\end{align}
where $u$ varies in $S^{2n-1}_{\vert M\vert}$. Then, for $M<0$, define
$L_{j,\ell}^{y,-}\colon\CP^{n-1}\to\C$ by
\begin{equation}
L_{j,\ell}^{y,-}\circ \pi_{\vert M\vert} (u)=L_{j,\ell}^{y,+}\circ \pi_{\vert M\vert}(\bar u),
\label{L^y-}
\end{equation}
referring to the right-hand-side of \eqref{L^y+}
with the understanding that now $y>p\pi/n$.
\end{definition}
Next, we prove that the matrices $L_y^\loc$ and $L^{y,\pm}\circ\pi_{|M|}\circ\cE$,
are similar and can be
transformed into each other by a unitary matrix. This is one of our main results.

\begin{theorem}
\label{theor:3.4}
The smooth matrix function $L^{y,\pm}\colon\CP^{n-1}\to\C^{n\times n}$ with components
$L_{j,\ell}^{y,\pm}$ given by \eqref{L^y+},\eqref{L^y-} satisfies the following identity
\begin{equation}
(L^{y,\pm}\circ\pi_{|M|}\circ\cE)(\xi,e^{\ri\theta})
=\Delta(e^{\ri\theta})^{-1}L_y^\loc(\xi,e^{\ri\theta})\Delta(e^{\ri\theta}),\quad
\forall(\xi,e^{\ri\theta})\in\cA_y^+\times\T^{n-1},
\label{L^y-circ-cE}
\end{equation}
where $\Delta(e^{\ri\theta})=\diag(\Delta_1,\dots,\Delta_n)\in\UN(n)$ with
\begin{equation}
\Delta_j=\exp\bigg(\ri\sum_{k=1}^{n-1}\Omega_{j,k}\theta_k\bigg),
\quad j=1,\dots,n-1,\quad \Delta_n=1.
\label{Delta}
\end{equation}
Consequently, $L^{y,\pm}(\pi_{\vert M\vert}( u))\in\SU(n)$ for every
$u\in S^{2n-1}_{\vert M\vert} $, and $L^{y,\pm}$ provides an extension
of the local Lax matrix $L_y^\loc$ \eqref{L_y^loc} to the global phase
space $\CP^{n-1}$.
\end{theorem}

\begin{proof}
The form of the local Lax matrix $L_y^\loc$ \eqref{L_y^loc} and Definitions
\ref{def:3.1} and \ref{def:3.3} show that \eqref{L^y-circ-cE} is equivalent
to the equations
\begin{equation}
\Delta_j=\begin{cases}
\Delta_{j+p}\rho_{j+p},&\text{if}\quad 1\leq j\leq n-p,\\
\Delta_{j+p-n}\rho_{j+p-n},&\text{if}\quad n-p<j\leq n.
\end{cases}
\label{Delta-recursion}
\end{equation}
The two sides of \eqref{Delta-recursion} can be written as exponentials of linear
combinations of the variables $\theta_k$ $(1\leq k\leq n-1)$. We next spell out the
relations that ensure the exact matching of the coefficients of the $\theta_k$ in
these exponentials. Plugging the components of $\Delta$ and $\rho$ into
\eqref{Delta-recursion}, the case $1\leq j<n-p$ gives
\begin{equation}
\begin{split}
\Omega_{j,j+p-1}&=\Omega_{j+p,j+p-1}+1,\quad(\text{coefficients of}\ \theta_{j+p-1})\\
\Omega_{j,j+p}&=\Omega_{j+p,j+p}-1,\quad(\text{coefficients of}\ \theta_{j+p})\\
\Omega_{j,k}&=\Omega_{j+p,k},\quad(\text{coefficients of}\ \theta_k,\ k\neq j+p-1,j+p),
\end{split}
\label{Omega-1}
\end{equation}
while for $j=n-p$ we get
\begin{equation}
\begin{split}
\Omega_{n-p,n-1}&=1,\quad(\text{coefficients of}\ \theta_{n-1})\\
\Omega_{n-p,k}&=0,\quad(\text{coefficients of}\ \theta_k,\ k\neq n-1).
\end{split}
\label{Omega-2}
\end{equation}
The case $n-p<j<n$ (and $p>1$) leads to
\begin{equation}
\begin{split}
\Omega_{j,j+p-n-1}&=\Omega_{j+p-n,j+p-n-1}+1,\quad(\text{coefficients of}\ \theta_{j+p-n-1})\\
\Omega_{j,j+p-n}&=\Omega_{j+p-n,j+p-n}-1,\quad(\text{coefficients of}\ \theta_{j+p-n})\\
\Omega_{j,k}&=\Omega_{j+p-n,k},\quad(\text{coefficients of}\ \theta_k,\ k\neq j+p-n-1,j+p-n).
\end{split}
\label{Omega-3}
\end{equation}
For $j=n$ there are two possibilities. If $p=1$ then we obtain
\begin{equation}
\begin{split}
\Omega_{1,1}&=1,\quad(\text{coefficients of}\ \theta_1)\\
\Omega_{1,k}&=0,\quad(\text{coefficients of}\ \theta_k,\ k\neq 1),
\end{split}
\label{Omega-4}
\end{equation}
and if $p>1$ then we require
\begin{equation}
\begin{split}
\Omega_{p,p-1}&=-1,\quad(\text{coefficients of}\ \theta_{p-1})\\
\Omega_{p,p}&=1,\quad(\text{coefficients of}\ \theta_p)\\
\Omega_{p,k}&=0,\quad(\text{coefficients of}\ \theta_k,\ k\neq p-1,p).
\end{split}
\label{Omega-5}
\end{equation}
Using the explicit formula given by Proposition \ref{prop:2.1},
we now show that $\Omega$ satisfies \eqref{Omega-1}.
Since $\Omega_{j,k}=B_{j,k}-C_{j,k}$ for all $j,k$, where $B_{j,k}$ \eqref{B_m,k}
depends on $(k-j)$ and $C_{j,k}$ \eqref{C_m,k} depends only on $k$, the equations
\eqref{Omega-1} reduce to
\begin{equation}
\begin{split}
B_{j,j+p-1}&=B_{j+p,j+p-1}+1,\\
B_{j,j+p}&=B_{j+p,j+p}-1,\\
B_{j,k}&=B_{j+p,k},\quad k\neq j+p-1,j+p.
\end{split}
\label{B-1}
\end{equation}
The first equation holds, because $(j+p-1)-j=p-1\equiv (n-q+1)p\pmod{n}$
implies $B_{j,j+p-1}=1$ and $(j+p-1)-(j+p)=-1\equiv (n-q)p\pmod{n}$ implies
$B_{j+p,j+p-1}=0$. For the second equation, we plainly have
$B_{j,j+p}=0$, and $(j+p)-(j+p)=0\equiv np\pmod{n}$ gives $B_{j+p,j+p}=1$.
Regarding the third equation, notice that $B_{j,k}=0$ in \eqref{B-1} when
$k-j\equiv\ell p\pmod{n}$ for some $\ell\in\{2,\dots,n-q\}$, and then
$B_{j+p,k}=0$ holds, too. Conversely, $B_{j+p,k}=0$ in \eqref{B-1} means
that $(k-j)-p\equiv\ell p\pmod{n}$ for some $\ell\in\{1,\dots,n-q-1\}$,
from which $(k-j)\equiv(\ell+1)p\pmod{n}$ and thus $B_{j,k}=0$ follows.
As $B$ is a $(0,1)$-matrix, we conclude that \eqref{B-1} is valid.
Proceeding in a similar manner, we have verified the rest of the relations
\eqref{Omega-2}--\eqref{Omega-5} as well. Since the relations
\eqref{Omega-1}--\eqref{Omega-5} imply \eqref{Delta-recursion},
the proof is complete.
\end{proof}

It is an immediate consequence of Theorem \ref{theor:3.4} that the spectral invariants of the global Lax matrix
$L^{y, \pm} \in C^\infty(\CP^{n-1}, \SU(n))$ yield a Liouville integrable system.
Because of \eqref{Re-tr-L_y^loc} the corresponding Poisson commuting family contains
the extension of the III$_\b$ Hamiltonian $H_y^\loc$ to $\CP^{n-1}$ for any type (i)
coupling. The self-duality of this compactified RS system was established in \cite{FKl},
and it will be studied in more detail elsewhere.

\section{Compact forms of the elliptic RS system}
\label{sec:4}

In this section we explain that type (i) compactifications of the elliptic RS system
can be constructed in exactly the same way as we saw for the trigonometric system.
This is due to the fact that the local elliptic Lax matrix is built from
the $\ws$-function \eqref{sigma} similarly as its trigonometric counterpart is built
from the sine function, and on the real axis these two functions have the same
zeros, signs, parity and antiperiodicity property.

We start by recalling some formulas of the relevant elliptic functions.
First, let $\omega,\omega'$ stand for the half-periods of the Weierstrass $\wp$
function defined by
\begin{equation}
\wp(z;\omega,\omega')
=\frac{1}{z^2}+\sum_{\substack{m,m'=-\infty\\(m,m')\neq(0,0)}}^\infty
\bigg[\frac{1}{(z-\omega_{m,m'})^2}-\frac{1}{\omega_{m,m'}^2}\bigg],
\label{wp}
\end{equation}
with $\omega_{m,m'}=2m\omega+2m'\omega'$. We adopt the convention
$\omega,-\ri\omega'\in(0,\infty)$, which ensures that $\wp$ is positive
on the real axis. Next, introduce the following `$\ws$-function':
\begin{equation}
\ws(z;\omega,\omega')
=\frac{2\omega}{\pi}\sin\!\Bigl(\frac{\pi z}{2\omega}\Bigr)
\prod_{m=1}^\infty\bigg[1+\frac{\sin^2(\pi z/(2\omega))}{\sinh^2(m\pi|\omega'| /\omega)}\bigg],
\label{sigma}
\end{equation}
related to the Weierstrass $\sigma$ and $\zeta$ functions by
$\ws(z)=\sigma(z)\exp(-\eta z^2/(2\omega))$ with the constant $\eta=\zeta(\omega)$.
A useful identity connecting $\wp$ and $\ws$ is
\begin{equation}
\frac{\ws(z+z')\ws(z-z')}{\ws^2(z)\ws^2(z')}=\wp(z')-\wp(z),
\qquad z, z'\in \C.
\label{wp-sigma}
\end{equation}
The $\ws$-function is odd, has simple zeros at $\omega_{m,m'}$ $(m,m'\in\Z)$ and
enjoys the scaling property $\ws(t z;t\omega,t\omega')=t\ws(z;\omega,\omega')$.
From now on we take
\begin{equation}
\omega=\frac{\pi}{2},
\end{equation}
whereby $\ws(z+\pi)=-\ws(z)$ holds as well. The trigonometric limit is obtained according to
\begin{equation}
\lim_{-\ri\omega'\to\infty}\wp(z;\pi/2,\omega')=\frac{1}{\sin^2(z)}-\frac{1}{3},\quad
\lim_{-\ri\omega'\to\infty}\ws(z;\pi/2,\omega')=\sin(z).
\label{trig-limit}
\end{equation}

Let us now pick a type (i) coupling parameter $y$ \eqref{typeI-y} and choose
the domain of the dynamical variables to be the same $\cA_y^+\times\T^{n-1}$
as in the trigonometric case. Then consider the following IV$_\b$ variant of
the standard \cite{RuijCMP87,RuijR} elliptic RS Lax matrix:
\begin{equation}
L_y^\loc(\xi,e^{\ri\theta}\vert\lambda)_{j,\ell}
=\frac{\ws(y)}{\ws(\lambda)}
\frac{\ws(x_j-x_\ell+\lambda)}{\ws(x_j-x_\ell+y)}
v_j(\xi,y)v_\ell(\xi,-y)\rho(\theta)_\ell,
\,\,\, \forall(\xi,e^{\ri\theta})\in\cA_y^+\times\T^{n-1},
\label{L_y^loc-IV}
\end{equation}
where $\lambda\in\C\setminus\{\omega_{m,m'}:m,m'\in\Z\}$ is a spectral parameter
and $v_\ell(\xi,\pm y)=\sqrt{z_\ell(\xi,\pm y)}$ with
\begin{equation}
z_\ell(\xi,\pm y)=\sgn(\ws(ny))\prod_{m=\ell+1}^{\ell+n-1}
\frac{\ws(\sum_{k=\ell}^{m-1}\xi_k\mp y)}{\ws(\sum_{k=\ell}^{m-1}\xi_k)}.
\label{z_ell(xi,pmy)-IV}
\end{equation}
These formulas are to be compared with the trigonometric case. Since $\ws(z)$ and
$\sin(z)$ have matching properties on the real line, we can repeat the arguments
presented in Section \ref{sec:3} to verify that $z_\ell(\xi,\pm y)>0$ for every
$\xi\in\cA_y^+$. Taking positive square roots, and applying the relation $x_{k+1}-x_k=\xi_k$
to express $x_j-x_\ell$ in terms of $\xi$, we conclude that the above local Lax matrix
is a smooth function on $\cA_y^+\times\T^{n-1}$ for every allowed value of the spectral
parameter. The fact that it is a specialization of the standard elliptic Lax matrix
ensures \cite{RuijCMP87,RuijR} that its characteristic polynomial generates $(n-1)$
independent real Hamiltonians in involution with respect to the symplectic form \eqref{om-loc}.
Indeed, the characteristic polynomial has the form
\begin{equation}
\det\big(L_y^\loc(\xi,e^{\ri\theta}\vert\lambda)-\alpha\1_n\big)
=\sum_{k=0}^n(-\alpha)^{n-k}c_k(\lambda,y)\cS_k^\loc(\xi,e^{\ri\theta},y),
\label{char1}\end{equation}
where the functions $\cS_k^\loc$ as well as their real and imaginary parts Poisson
commute, and $\Re(\cS_{k}^\loc)$ for $k=1,\dots, n-1$ are functionally independent.
Explicit formulas of the $c_k$ (that do not depend on the phase space variables)
and $\cS_k^\loc$ (that do not depend on $\lambda$) can be found in \cite{RuijCMP87,RuijR}.
The function $\Re(\cS^\loc_1)$ is the RS Hamiltonian of IV$_\b$ type
\begin{equation}
\Re\big(\tr L_y^\loc(\xi,e^{\ri\theta}\vert\lambda)\big)
=\sum_{j=1}^n\cos(\theta_j-\theta_{j-1})\sqrt{ \prod_{m=j+1}^{j+n-1}
\left[\ws(y)^2(\wp(y)-\wp(\textstyle\sum_{k=j}^{m-1}\xi_k))\right]}.
\label{H_y^loc-IV}
\end{equation}

We note in passing that in Ruijsenaars's papers \cite{RuijCMP87,RuijR} one finds
the elliptic Lax matrix $VL_y^\loc V^{-1}$, where $V$ is the diagonal matrix
$V=\rho(\theta)\diag(v_1(\xi,-y),\dots,v_n(\xi,-y))$. This difference is
irrelevant, since it has no effect on the generated spectral invariants.
Another difference is that we work in the center-of-mass frame.

Now the complete train of thought applied in the previous section
remains valid if we simply replace the sine function with the $\ws$-function everywhere.
In particular, the direct analogues of the formulas \eqref{v_ell(xi,y)}--\eqref{w_ell(xi,-y)-2}
hold with smooth functions $w_k(\xi,\pm y)>0$, for $\xi\in\cA_y$. Due to this fact,
we can introduce a smooth elliptic Lax matrix defined on the global phase space $\CP^{n-1}$.
The subsequent definition refers to the explicit formulas of Appendix \ref{sec:A},
which in the elliptic case contain the function
\begin{equation}
\cJ(|u_k|^2)=\frac{\ws(|u_k|^2)}{|u_k|^2}.
\label{J2}\end{equation}
This has the same smoothness and positivity properties at and around zero as $J$ \eqref{J1} does.
We also use $\xi(u)$ \eqref{xi(u)} and the functions $(x_j-x_\ell)(\xi)$ determined by $x_{k+1}-x_k=\xi_k$.

\begin{definition}
\label{def:4.1}
Take a type (i) $y$ from \eqref{typeI-y} and represent the points of $\CP^{n-1}$ as
$\pi_{|M|}(u)$ with $u\in S^{2n-1}_{|M|}$. For $M>0$ \eqref{M}, define the smooth
functions $\cL_{j,\ell}^{y,+}$ on $\CP^{n-1}$ by
\begin{align}
\cL_{j,\ell}^{y,+}(\pi_{\vert M\vert}(u))&=\begin{cases}
\Lambda_{j,j+p}^y(\xi(u)),&\text{if}\ 1\leq j\leq n-p,\ \ell=j+p,\\
\Lambda_{j,j+p-n}^y(\xi(u)),&\text{if}\ n-p<j\leq n,\ \ell=j+p-n,\\
\bar u_ju_{\ell+n-p}\Lambda_{j,\ell}^y(\xi(u)),&\text{if}\ 1\leq j\leq n,\ 1\leq\ell\leq p,\ \ell\neq j+p-n,\\
\bar u_ju_{\ell-p}\Lambda_{j,\ell}^y(\xi(u)),&\text{if}\ 1\leq j\leq n,\ p<\ell\leq n,\ \ell\neq j+p,
\end{cases}
\label{L^y+IV}
\end{align}
with $\Lambda_{j,\ell}^y$ given in Appendix \ref{sec:A}. For $M<0$, set $\cL_{j,\ell}^{y,-}$
to be
\begin{equation}
\cL_{j,\ell}^{y,-}(\pi_{\vert M\vert}(u))=\cL_{j,\ell}^{y,+}(\pi_{\vert M\vert}(\bar u))
\label{L^y-IV}
\end{equation}
with the understanding that in this case $y>p\pi/n$. Finally,
define the $\lambda$-dependent elliptic Lax matrix $L^{y,\pm}$ on $\CP^{n-1}$ by
\begin{equation}
L_{j,\ell}^{y,\pm}(\pi_{\vert M\vert}(u)\vert\lambda)=
\frac{\ws((x_j-x_\ell)(\xi(u))+\lambda)}{\ws(\lambda)}
\cL_{j,\ell}^{y,\pm}(\pi_{\vert M\vert}(u)),
\label{L^y}
\end{equation}
where $u$ runs over $S^{2n-1}_{\vert M\vert}$ and the spectral parameter $\lambda$
varies in $\C\setminus\{\omega_{m,m'}:m,m'\in\Z\}$.
\end{definition}

\begin{theorem}
\label{theor:4.2}
The spectral parameter dependent elliptic Lax matrix $L^{y,\pm}(\pi_{\vert M\vert}(u)\vert\lambda)$
\eqref{L^y} is a smooth global extension of $L_y^\loc(\xi,e^{\ri \theta}\vert\lambda)$
\eqref{L_y^loc-IV} to the complex projective space $\CP^{n-1}$ since it satisfies
\begin{equation}
L^{y,\pm}((\pi_{|M|}\circ \cE)(\xi,e^{\ri\theta})\vert\lambda)
=\Delta(e^{\ri\theta})^{-1}L_y^\loc(\xi,e^{\ri\theta}\vert\lambda)
\Delta(e^{\ri\theta}),\quad
\forall(\xi,e^{\ri\theta})\in\cA_y^+\times\T^{n-1},
\label{L^y-circ-pi_M-circ-cE}
\end{equation}
where $\Delta$ is given by \eqref{Delta} and
$\pi_{\vert M\vert}\circ\cE\colon\cA_y^+\times\T^{n-1}\to \CP^{n-1}$
is the symplectic embedding defined in Section \ref{sec:2}.
\end{theorem}

The proof of Theorem \ref{theor:4.2} follows the lines of the proof of Theorem \ref{theor:3.4}.
The characteristic polynomial $\det\big(L^{y,\pm}(\pi_{\vert M\vert}(u)\vert\lambda)-\alpha\1_n\big)$
of the global Lax matrix depends smoothly on $\pi_{\vert M\vert}(u)\in \CP^{n-1}$ and as a
consequence of \eqref{L^y-circ-pi_M-circ-cE} it satisfies
\begin{equation}
\det\big(L^{y,\pm}((\pi_{|M|}\circ\cE)(\xi,e^{\ri\theta})\vert\lambda)
-\alpha \1_n\big)=
\det\big(L_y^\loc(\xi,e^{\ri\theta}\vert\lambda)-\alpha\1_n\big).
\end{equation}
Since this holds for all $\alpha$ and $\lambda$, we see that the local IV$_\b$
Hamiltonian \eqref{H_y^loc-IV} together with its constants of motion $\Re(\cS_k^\loc)$,
$k=2,\dots,n-1$ extends to an integrable system on $\CP^{n-1}$. This was pointed out
previously \cite{RuijR} for the special case $0<y<\pi/n$ in \eqref{typeI-y}.

In the trigonometric limit $-\ri\omega'\to\infty$ the $\ws$-function becomes the sine function,
and we obtain a spectral parameter dependent trigonometric Lax matrix from the elliptic one.
Then, setting the spectral parameter to be on the imaginary axis and taking the limit
$-\ri\lambda\to\infty$ reproduces, up to conjugation by a diagonal matrix, the trigonometric
global Lax matrix of Definition \ref{def:3.3}. Correspondingly, the global extension of the
IV$_\b$ Hamiltonian \eqref{H_y^loc-IV} and its commuting family reduces to the global extension
of the III$_\b$ Hamiltonian \eqref{H_y^loc} and its constants of motion.

\section{Conclusion and outlook}
\label{sec:5}

In this paper we have demonstrated by direct construction that the local phase space
$\cA_y^+\times\T^{n-1}$ of the III$_\b$ and IV$_\b$ RS models (where $\cA_y^+$ is the
interior of the simplex \eqref{3.11}) can be embedded into $\CP^{n-1}$ for any type (i)
coupling $y$ \eqref{typeI-y} in such a way that a suitable conjugate of the local Lax
matrix extends to a smooth (actually real-analytic) function. Theorems \ref{theor:3.4}
and \ref{theor:4.2} together with Appendix \ref{sec:A} provide explicit formulas for
the resulting global Lax matrices. Their characteristic polynomials give rise to Poisson
commuting real Hamiltonians on $\CP^{n-1}$ that yield the Liouville integrable compactified
trigonometric and elliptic RS systems.

Our direct construction was inspired by the earlier derivation of compactified
III$_\b$ systems by quasi-Hamiltonian reduction \cite{FKl}. The reduction identifies
the III$_\b$ system with a topological Chern-Simons field theory for any generic
coupling parameter $y$. It appears natural to ask if an analogous derivation and
relation to some topological field theory could exist for IV$_\b$ systems, too.
We also would like to obtain a better understanding of the type (ii) trigonometric
systems and their possible elliptic analogues.

In the near future, we wish to explore the classical dynamics and quantization of
the III$_\b$ systems. This is partially motivated by the possibility to associate
new random matrix ensembles with these systems \cite{Bog}. For arbitrary type (i) couplings, geometric
quantization yields the joint spectra of the quantized action variables effortlessly
\cite{FK2}. (It is necessary to introduce a second parameter into the systems before
quantization, which can be achieved by taking an arbitrary multiple of the symplectic form.)
The joint eigenfunctions of the quantized RS Hamiltonian and its commuting family
should be derived by generalizing the results of van Diejen and Vinet \cite{vDV}.

Besides further studying the systems that we described, it would be also interesting
to search for compactifications of generalized RS systems. We have in mind especially
the $\mathrm{BC}_n$ systems due to van Diejen \cite{vD} and the recently introduced
supersymmetric systems \cite{susy}. Regarding the former case, and even for general
root systems, the results of \cite{vDE} could be relevant, as well as the construction
of Lax matrices for some of the $\mathrm{BC}_n$ systems reported in \cite{PG}.

Throughout the text, we worked in the `center-of-mass frame' and now we end by a
comment on how the center-of-mass coordinate can be introduced into our systems.
One possibility is to take the full phase space to be the Cartesian product
of $\CP^{n-1}$ with $\UN(1)\times\UN(1)=\{(e^{2\ri X},e^{\ri\Phi})\}$
endowed with the symplectic form $|M|\,\omega_{\mathrm{FS}}+dX\wedge d\Phi$.
Here, $e^{2\ri X}$ is interpreted as a center-of-mass variable for the $n$
particles on the circle. Then $n$ functions in involution result by adding an
arbitrary function of $e^{\ri\Phi}$ to the $(n-1)$ commuting Hamiltonians
generated by the `total Lax matrix' $e^{-\ri \Phi}L^{y,\pm}$. On the dense open
domain the total Lax matrix is obtained by replacing $\rho(\theta)$ in \eqref{L_y^loc-IV}
by $\rho(\theta)e^{-\ri\Phi}$. By setting $e^{\ri\Phi}$ to $1$ and quotienting
by the canonical transformations generated by the functions of $e^{\ri\Phi}$ one
recovers the phase space of the relative motion, $\CP^{n-1}$. There are also several
other possibilities, as was discussed for analogous situations in \cite{RIMS95,FA}.
For example, one may replace $\UN(1)\times\UN(1)$ by its covering space $\R\times\R$.

\bigskip
\begin{acknowledgements}
We thank B.G.~Pusztai for helpful comments and S.~Ruijsenaars for useful
discussions.
This work was supported in part by the Hungarian Scientific Research Fund (OTKA)
under the grant K-111697 and by COST (European Cooperation in Science and Technology)
in COST Action MP1405 QSPACE.
\end{acknowledgements}

\newpage
\appendix

\section{Explicit form of the functions $\boldsymbol{\Lambda_{j,\ell}^y}$}
\label{sec:A}

In this appendix we display the building blocks \eqref{L^y+IV} of the global
elliptic Lax matrix explicitly. Below, $\xi$ varies in the closed simplex $\cA_y$
associated with a type (i) coupling $y$ \eqref{typeI-y} for fixed $p$ and $M$.
The function $\cJ$ was defined in \eqref{J2}. The trigonometric case is obtained
by simply replacing the $\ws$-function \eqref{sigma} everywhere by the sine function.

\smallskip
\noindent\underline{Special components:}
For $1\leq j\leq n-p$
$$
\Lambda_{j,j+p}^y(\xi)=-\sgn(M)\ws(y)
\frac{\big[\prod_{\substack{m=1\\(m\neq p)}}^{n-1}
\ws(\sum_{k=j}^{j+m-1}\xi_k-y)
\ws(\sum_{k=j+p}^{j+p+n-m-1}\xi_k+y)\big]^{\tfrac{1}{2}}}
{\prod_{m=1}^{n-1}\big[\ws(\sum_{k=j}^{j+m-1}\xi_k)
\ws(\sum_{k=j+p}^{j+p+m-1}\xi_k)\big]^{\tfrac{1}{2}}}.
$$
For $n-p<j\leq n$
$$
\Lambda_{j,j+p-n}^{y}(\xi)=\sgn(M)\ws(y)
\frac{\big[\prod_{\substack{m=1\\(m\neq p)}}^{n-1}
\ws(\sum_{k=j}^{j+m-1}\xi_k-y)
\ws(\sum_{k=j+p-n}^{j+p-m-1}\xi_k+y)\big]^{\tfrac{1}{2}}}
{\prod_{m=1}^{n-1}\big[\ws(\sum_{k=j}^{j+m-1}\xi_k)
\ws(\sum_{k=j+p-n}^{j+p-m-1}\xi_k)\big]^{\tfrac{1}{2}}}.
$$
\underline{Diagonal components:}
For $1\leq j=\ell\leq p$
$$
\Lambda_{j,j}^y(\xi)=\big[\cJ(|u_j|^2)\cJ(|u_{j+n-p}|^2)\big]^{\tfrac{1}{2}}
\frac{\big[\prod_{\substack{m=1\\(m\neq p)}}^{n-1}
\ws(\sum_{k=j}^{j+m-1}\xi_k-y)
\ws(\sum_{k=j}^{j+n-m-1}\xi_k+y)\big]^{\tfrac{1}{2}}}
{\prod_{m=1}^{n-1}\ws(\sum_{k=j}^{j+m-1}\xi_k)}.
$$
For $p<j=\ell\leq n$
$$
\Lambda_{j,j}^y(\xi)=\big[\cJ(|u_j|^2)\cJ(|u_{j-p}|^2)\big]^{\tfrac{1}{2}}
\frac{\big[\prod_{\substack{m=1\\(m\neq p)}}^{n-1}
\ws(\sum_{k=j}^{j+m-1}\xi_k-y)
\ws(\sum_{k=j}^{j+n-m-1}\xi_k+y)\big]^{\tfrac{1}{2}}}
{\prod_{m=1}^{n-1}\ws(\sum_{k=j}^{j+m-1}\xi_k)}.
$$
\underline{Components above the diagonal:}
For $1\leq j<\ell\leq p$
$$
\Lambda_{j,\ell}^y(\xi)
=\ws(y)\big[\cJ(|u_j|^2)\cJ(|u_{\ell+n-p}|^2)\big]^{\tfrac{1}{2}}
\frac{\big[\prod_{\substack{m=1\\(m\neq\ell-j,p)}}^{n-1}
\ws(\sum_{k=j}^{j+m-1}\xi_k-y)
\ws(\sum_{k=\ell}^{\ell+n-m-1}\xi_k+y)\big]^{\tfrac{1}{2}}}
{\prod_{m=1}^{n-1}\big[\ws(\sum_{k=j}^{j+m-1}\xi_k)
\ws(\sum_{k=\ell}^{\ell+m-1}\xi_k)\big]^{\tfrac{1}{2}}}.
$$
For $1\leq j<\ell\leq n$ with $p<\ell$ and $\ell\neq j+p$
$$
\Lambda_{j,\ell}^y(\xi)=\frac{\ws(y)\big[\cJ(|u_j|^2)
\cJ(|u_{\ell-p}|^2)\big]^{\tfrac{1}{2}}}{\sgn(j+p-\ell)}
\frac{\big[\prod_{\substack{m=1\\(m\neq\ell-j,p)}}^{n-1}
\ws(\sum_{k=j}^{j+m-1}\xi_k-y)
\ws(\sum_{k=\ell}^{\ell+n-m-1}\xi_k+y)\big]^{\tfrac{1}{2}}}
{\prod_{m=1}^{n-1}\big[\ws(\sum_{k=j}^{j+m-1}\xi_k)
\ws(\sum_{k=\ell}^{\ell+m-1}\xi_k)\big]^{\tfrac{1}{2}}}.
$$
\underline{Components below the diagonal:}
For $1\leq\ell<j\leq n$ with $\ell\leq p$ and $\ell\neq j+p-n$
$$
\Lambda_{j,\ell}^y(\xi)
=\frac{\ws(y)\big[\cJ(|u_j|^2)
\cJ(|u_{\ell+n-p}|^2)\big]^{\tfrac{1}{2}}}{\sgn(\ell+n -j- p )}
\frac{\big[\prod_{\substack{m=1\\(m\neq j-\ell,p)}}^{n-1}
\ws(\sum_{k=j}^{j+n-m-1}\xi_k-y)
\ws(\sum_{k=\ell}^{\ell+m-1}\xi_k+y)\big]^{\tfrac{1}{2}}}
{\prod_{m=1}^{n-1}\big[\ws(\sum_{k=j}^{j+m-1}\xi_k)
\ws(\sum_{k=\ell}^{\ell+m-1}\xi_k)\big]^{\tfrac{1}{2}}}.
$$
For $p<\ell<j\leq n$
$$
\Lambda_{j,\ell}^y(\xi)=\ws(y)\big[\cJ(|u_j|^2)
\cJ(|u_{\ell-p}|^2)\big]^{\tfrac{1}{2}}
\frac{\big[\prod_{\substack{m=1\\(m\neq j-\ell,p)}}^{n-1}
\ws(\sum_{k=j}^{j+n-m-1}\xi_k-y)
\ws(\sum_{k=\ell}^{\ell+m-1}\xi_k+y)\big]^{\tfrac{1}{2}}}
{\prod_{m=1}^{n-1}\big[\ws(\sum_{k=j}^{j+m-1}\xi_k)
\ws(\sum_{k=\ell}^{\ell+m-1}\xi_k)\big]^{\tfrac{1}{2}}}.
$$

\end{document}